\documentclass[12pt]{iopart}

\usepackage{dcolumn}
\usepackage{bm}
\usepackage{verbatim}       

\usepackage[dvips]{graphicx}
\usepackage{amsthm}
\usepackage{amssymb}
\usepackage{bm}
\usepackage{amstext}
\usepackage{psfrag}

\expandafter\let\csname equation*\endcsname\relax

\expandafter\let\csname endequation*\endcsname\relax

\usepackage{amsmath}
\usepackage{amsfonts}
\usepackage{mathrsfs}
\usepackage{cite}

\newcommand{\p}{\partial}

\newcommand{\dd}{{\rm d}}

\newcommand{\bd}{\begin{definition}}                
\newcommand{\ed}{\end{definition}}                  
\newcommand{\bc}{\begin{corollary}}                 
\newcommand{\ec}{\end{corollary}}                   
\newcommand{\bl}{\begin{lemma}}                     
\newcommand{\el}{\end{lemma}}                       
\newcommand{\bp}{\begin{proposition}}            
\newcommand{\ep}{\end{proposition}}                
\newcommand{\bere}{\begin{remark}}                  
\newcommand{\ere}{\end{remark}}                     

\newcommand{\bt}{\begin{theorem}}
\newcommand{\et}{\end{theorem}}

\newcommand{\bit}{\begin{itemize}}
\newcommand{\eit}{\end{itemize}}
\newtheorem{theorem}{Theorem}[section]
\newtheorem{corollary}[theorem]{Corollary}
\newtheorem{lemma}[theorem]{Lemma}
\newtheorem{proposition}[theorem]{Proposition}
\theoremstyle{definition}
\newtheorem{definition}[theorem]{Definition}
\theoremstyle{remark}
\newtheorem{remark}[theorem]{Remark}

\begin{document}

\title{Raychaudhuri equation and singularity theorems in Finsler spacetimes}


\author{E Minguzzi}
\address{Dipartimento di Matematica e Informatica ``U. Dini'', Universit\`a
degli Studi di Firenze, Via S. Marta 3,  I-50139 Firenze, Italy.}

\ead{ettore.minguzzi@unifi.it}

\date{}


\begin{abstract}
\noindent The Raychaudhuri equation and its consequences for chronality are studied in the context of  Finsler spacetimes. It is proved that  the notable singularity theorems of Lorentzian geometry extend to the Finslerian domain. Indeed, so do the theorems by  Hawking, Penrose, Hawking and Penrose, Geroch, Gannon, Tipler, or Kriele, but also the Topological Censorship theorem  and so on. It is argued that the notable results in causality theory connected to achronal sets, future sets, domains of dependence, limit curve theorems, length functional, Lorentzian distance, geodesic connectedness, extend to the Finslerian domain. Results concerning the spacetime asymptotic structure, horizons differentiability and conformal transformations are also included.
\end{abstract}


\section{Introduction}
The prediction of singularities is one of the most interesting and surprising features of general relativity.  Remarkably  this theory is able to signal through these results its own limits of application. We refer the reader to \cite{geroch68,tipler80,canarutto88,senovilla97} for excellent discussions on the physical interpretation of singularity theorems.

The purpose of this work is to show that all the notable singularity theorems of Lorentzian geometry extend to the Finslerian domain.  The translation of these results from the Lorentzian to the Finslerian domain will be mostly word for word provided some   preliminary work is done in order to build the correspondence and clarify some concepts, from the definition of expansion of a congruence, to the definition of the genericity  and energy conditions. Of course, we will have first to obtain and study  a Finslerian Raychaudhuri equation. The reader may find another  derivation with application to cosmology in \cite{stavrinos09,stavrinos12,kouretsis14,basilakos13}.


Our philosophy will be that of avoiding excessive use of linear Finsler connections.
We recall that in Finsler geometry one has both a non-linear connection on the bundle $TM\to M$ and many notable Finsler (linear) connections on the vertical bundle $V(TM\backslash 0) \to TM\backslash 0$, say Berwald, Cartan, Chern-Rund and Hashiguchi. Furthermore, given a section $s\colon M\to TM\backslash 0$, one has  four more notable connections obtained by  pulling back the Finsler connections to $TM\to M$ (for a study of the pullback connection see \cite{ingarden93}), to which we can add the Levi-Civita connection $\nabla^{s^*g}$ of the pullback metric $s^*g$.
 Among these connections the non-linear connection is more primitive, in fact all the notable Finsler connections induce the same non-linear connection. Since we wish to show  that causality theory for Finslerian spacetimes is completely independent of the Finsler connection used, we shall use (linear) Finsler connections minimally and only as tools in proofs.
  For a recent investigation of  Penrose's theorem in Finsler geometry
 written with a different philosophy the reader might consult \cite{aazami14}. For other papers related to causality in Finsler spacetimes the reader is referred to \cite{beem70,beem74,ishikawa81,perlick06,skakala09,gallego12,lammerzahl12,pfeifer11,pfeifer12,minguzzi13c,minguzzi13d,kouretsis14} while for the dynamics of Finsler gravity are also interesting \cite{ishikawa80,miron92,rutz93,gibbons07,voicu10,castro12,silagadze11,li14}.

We wish to mention one difficulty we met in the generalization of singularity theorems.
These theorems use results which relate the existence of conjugate (focusing) points over null geodesics (resp.\ hypersurfaces) with their chronality.
The classical proofs of this fact, based on the study of curve variations, are quite technical and are, saved for Beem et al.\ proof \cite[Theor.\ 10.72]{beem96}, somewhat incomplete. Galloway has developed a cleaner version \cite{galloway96b} (the causality lemma) which, however, is still algebraic and somewhat technical. It can be expected that the extension of these proofs to the Finslerian domain could only  make them longer and more involved.
Fortunately, we shall obtain a different and rather simple topological proof, which we believe to be interesting already in the Lorentzian case. It does not involve neither curve variations nor the Morse index theory for geodesics.

 The reader will also find many other  Finslerian causality results which show that several other ingredients entering the classical Lorentzian proofs admit a Finslerian formulation.

We end this introductory section with some considerations on the present status of Finsler gravity theory.

While in recent years there have been improvements in the geometrical understanding of the theory, there is still much to be done in the study of dynamics. In Lorentz-Finsler geometry many
tensors reduce to Einstein's for metrics independent of velocity, while no choice seems to imply a satisfactory conservation law over $M$. Either further conditions need to be imposed, or the conservation should be understood over $TM$  \cite{rund62,shimada77,ishikawa80,miron87,ikeda81,anastasiei87,minguzzi14c}. As a consequence, the dynamical equations are less constrained than one might have hoped for.

All authors seem to agree on the validity of the vacuum equation (see next section for notations)
\begin{equation}
\textrm{Ric}(v)=0,
\end{equation}
indeed, it is implied by almost every choice of dynamical equations that has been proposed so far, starting from the first proposal by Horvath \cite{horvath50,horvath52}. Not all authors obtained this equation from a tensorial generalization of Einstein's. Rutz \cite{rutz93}, for instance, argued for its validity using an analogy based on the the Jacobi deviation equation (cf.\ Eq. (\ref{cmo})). The reader might equivalently infer its validity from the form of the Raychadhuri equation (\ref{ray}), to be derived in the next sections. While the dynamical equations of general relativity are uniquely determined by the conservation of energy-momentum and a few other assumptions, in Finsler geometry the vacuum equation $\textrm{Ric}(v)=0$ looks more like a reasoned guess. Consensus has yet to be reached on its generalization in presence of matter, and furthermore,
several authors suggested to impose further equations in order to constrain the vertical degrees of freedom \cite{ishikawa80,miron87} (e.g.\ we argued in favor of a vanishing mean Cartan torsion \cite{minguzzi14c}).

Researchers have not been discouraged from looking at the dynamics of the theory (often changing dynamical equations from  work to work), though a study of the Cauchy problem has yet to be performed. Indeed, many authors have given Finslerian generalizations of the Schwarzschild's metric or of the Friedmann-Robertson-Walker metric \cite{asanov92,rutz93,lammerzahl12,basilakos13,li14}. Others have investigated the effect on the dispersion relations \cite{girelli07,lammerzahl09}. They started from a metric ansatz dependent on some anisotropic parameters or functions. Clearly, since any dynamical equation should reduce to Einstein's in the quadratic Lagrangian case, the usual Lorentzian solutions solve the Finslerian dynamical equations. Thus these solutions pass the standard observational tests for sufficiently small anisotropic parameters.



We were motivated to study Finsler gravity by the following argument concerning the small scale structure  of spacetime. If the smoothness of spacetime is just an emerging feature, as it is widely held, then its mathematical structure should be that of a (quasi-pseudo)metric space, while the identification of a Riemannian metric would be just the result of an average or an observational phenomenon.
%
%
%
It is known (Busemann-Mayer's theorem \cite{tamassy08}) that a  metric on a sufficiently regular manifold determines a Finsler Lagrangian rather than a Riemannian metric, thus Finsler gravity should be expected as a  manifestation of the small scale (quantum) features of spacetime itself.

Physically, several questions remain to be answered: what is the source of anisotropy? Is it just a remnant of an anisotropic condition at beginning of the Universe, or there are specific anisotropic sources? Should the anisotropy be expected at microscopic scales as we suggested, the isotropy detected in  laboratory experiments being just the result of observational averaging, or should it also  be searched at large, possibly cosmological scales? In order to answer these questions it is necessary to determine the theory as completely as possible, proceeding on firm ground. Causality theory gives us this opportunity since it allows us to explore the theory resting on geometrical objects of physical significance.

\section{Elements of Lorentz-Finsler geometry} \label{sec}

The purpose of this section is mainly that of fixing  notation and terminology but can also serve as a fast introduction to Finsler geometry. Unfortunately, different schools have developed different conventions, thus we shall
give some key coordinate expressions to allow the reader to make fast correspondences with notations he or her might be used to. Of course, the objects introduced below can be given coordinate-free formulations; for those and for other introductions to Finsler geometry the reader is referred to \cite{antonelli93,abate94,bao00,shen01,dahl06,szilasi14,minguzzi14c}.

Let $M$ be a paracompact, Hausdorff, connected,\footnote{Subsequently ``$n$'' will also be used to denote a lightlike vector field, we hope that this fact will not generate confusion.} $n$+1-dimensional manifold. Let $\{x^\mu\}$ denote a local chart on $M$ and let $\{ x^\mu,v^\nu\}$ be the induced local chart on $TM$.
The Finsler Lagrangian is a  function on the slit tangent bundle $\mathscr{L}\colon TM\backslash 0 \to \mathbb{R}$ positive homogeneous of degree two in the velocities, $\mathscr{L}(x,sv)=s^2 \mathscr{L}(x,v)$ for every $s>0$.  The metric is defined as the Hessian  of $\mathscr{L}$ with respect to the velocities
\begin{equation}
g_{\mu \nu}(x,v)= \frac{\p^2 \mathscr{L}}{\p v^\mu \p v^\nu},
\end{equation}
and in index free notation will be also denoted  $g_v$ to stress the dependence on the velocity. This Finsler metric provides a map $g\colon TM\backslash 0 \to  T^*M \otimes T^*M$.

Lorentz-Finsler geometry is obtained whenever $g_v$ is Lorentzian, namely of signature $(-,+,\cdots,+)$ (in the Finsler, hence positive definite case, one often works with a function $F$, $\mathscr{L}=F^2/2$, in place of $\mathscr{L}$).
The just given definition of Lorentz-Finsler manifold is due to John Beem \cite{beem70}. We note that it is particularly convenient to work with a Lagrangian defined on the slit bundle $TM\backslash 0$ since the theory of Finsler connections traditionally has been developed on this space. Although the results of causality theory depend just on the Lagrangian restricted to the causal cone $\mathscr{L}\le 0$, such restriction is not needed since it does not bring more generality  and spoils the existence of convex neighborhoods. If one is given a Lagrangian defined just over the future causal cone it convenient, as a first step, to extend it on the whole slit tangent bundle, see \cite{minguzzi14h} for a complete discussion.

Let us recall some elements on the geometry of pseudo-Finsler connections (the reader is referred to \cite{minguzzi14c}). The Finsler Lagrangian allows us to define the geodesics as the stationary points  of the functional $\int \mathscr{L}(x,\dot x)\dd t$. The Lagrange equations are of second order and it turns out that a good starting point for the introduction of the Finsler connections is the notion of  {\em spray}.

We recall that a spray over $M$ can be locally characterized  as  second order differential equation
\[
\ddot x^\alpha+2G^\alpha(x,\dot x)=0,
\]
where $G^\alpha$ is positive homogeneous of degree two: $G^\alpha(x,s v)=s^2 G^\alpha(x,v)$ for every $s>0$. Let $E=TM\backslash 0$, and let $\pi_M\colon E\to M$ be the usual projection. This projection determines a vertical space $V_e E$ at every point $e\in E$. A non-linear connection is a splitting of the tangent space $TE=VE\otimes HE$ into vertical and horizontal bundles (for an introduction to the notion of non-linear connection the reader is referred to \cite{modugno91,michor08}). A basis for the horizontal space is given by
\[
\Big\{\frac{\delta \ }{\delta x^\mu}\Big\}, \qquad \frac{\delta}{\delta x^\mu}=\frac{\p}{\p x^\mu}-N^\nu_\mu(x,v) \frac{\p}{\p v^\nu},
\]
where the coefficients $N^\nu_\mu(x,v)$ define the non-linear connection and have suitable transformation properties under change of coordinates.
The curvature of the non-linear connection measures the non-holonomicity of the horizontal distribution
\begin{equation} \label{nsp}
\left[ \frac{\delta}{\delta x^\alpha}, \frac{\delta}{\delta x^\beta}\right]= -R^\mu_{\alpha \beta} \frac{\p}{\p v^\mu}, \qquad R^\mu_{\alpha \beta}=\frac{\delta N^\mu_\beta}{\delta x^\alpha}-\frac{\delta N^\mu_\alpha}{\delta x^\beta}.
\end{equation}
We can define the covariant derivative for the non-linear connection as follows. Given a section $s\colon U\to E$, $U\subset M$,
\[
 D_{\xi} s^\alpha=\Big(\frac{\p s^\alpha}{\p x^\mu} +N^\alpha_\mu\big(x, s(x)\big) \Big)\xi^\mu .
\]
The flipped derivative\footnote{
Unfortunately, some authors call it covariant derivative \cite{shen01}, while this term should be reserved to $D$.} is instead
\[
{\tilde{D}_{\xi}} s^\alpha=\frac{\p s^\alpha}{\p x^\mu} \, \xi^\mu+N^\alpha_\mu(x, \xi) s^\mu .
\]
and although well defined {\em is not} a covariant derivative in the standard sense since it is non-linear in the derivative vector $\xi$.
 Observe that if $X,Y\colon M\to TM$ are vector fields then
 \begin{equation} \label{kkq}
 {\tilde D_{X}} Y- D_Y X=[X,Y].
  \end{equation}
A geodesic is a curve $x(t)$ which satisfies $D_{\dot x} \dot{x}=0$, or equivalently ${\tilde D_{\dot x}} \dot x=0$.
We shall only be interested on the non-linear connection determined by a spray as follows \[N^\mu_\alpha=G^\mu_\alpha:=\p G^\mu/\p v^\alpha.\] The geodesics of this non-linear connection coincide with the integral curves of the spray. Furthermore, the geodesics of the spray will be  the stationary point of the action functional $\int\! \mathscr{L}\dd t$ thus
\begin{align}
2 {G}^\alpha(x,v)&= g^{\alpha\delta}\Big( \frac{\p^2\mathscr{L} }{\p x^\gamma \p v^\delta} \,v^\gamma -\frac{\p\mathscr{L} }{\p x^\delta } \Big) \label{axo}\\
&=\frac{1}{2}\, g^{\alpha \delta} \Big( \frac{\p}{\p x^\beta} \,g_{\delta \gamma}+\frac{\p}{\p x^\gamma} \, g_{\delta \beta}-\frac{\p}{\p x^\delta}\, g_{\beta \gamma}\Big) v^\beta v^\gamma . \label{axu}
\end{align}

It can be mentioned that in Finsler geometry one can define the linear Finsler connection $\nabla$, namely splittings of the vertical bundle $\pi_E\colon VE\to E$, $E=TM\backslash 0$. The  Berwald, Cartan, Chern-Rund and Hashiguchi connections are of this type. They are referred as {\em notable} Finsler connections.  Although different, they are all compatible with the same non-linear connection. In fact, the covariant derivative $X\to \nabla_X L$ of the Liouville vector field $L\colon E\to VE$, $L=v^\alpha\p/\p v^\alpha$, vanishes precisely over a $n$+1-dimensional distribution which determines a non-linear connection. This distribution is the same for all these connections and is determined by the spray as mentioned.

Each Finsler connection $\nabla$ determines two covariant derivatives $\nabla^H$ and $\nabla^V$ respectively being obtained from $\nabla_{\check X}$ whenever $\check X$ is the horizontal or the vertical lift of a vector $X\in TM$. In particular $\nabla^H$ is determined by local connection coefficients $H^\alpha_{\mu \nu}(x,v)$ which are related to those of the non-linear connection by (regularity) $N^\alpha_\mu(x,v)=H^\alpha_{\mu \nu}(x,v)v^\nu$. Examples are the Berwald connection
\[
H^\alpha_{\mu \nu}:=G^\alpha_{\mu \nu}:=\frac{\p}{\p v^\nu}\,G^\alpha_\mu,
\]
or the Chern-Rund or Cartan connection for both of which $\nabla^H g=0$ and hence
\begin{equation} \label{sog}
H^\alpha_{\beta \gamma}:=\Gamma_{\beta \gamma}^{\alpha}:=\frac{1}{2} g^{\alpha \sigma} \Big( \frac{\delta}{\delta x^\beta} \,g_{\sigma \gamma}+\frac{\delta}{\delta x^\gamma} \, g_{\sigma \beta}-\frac{\delta}{\delta x^\sigma}\, g_{\beta \gamma}\Big).
\end{equation}
The difference
\begin{equation} \label{lan}
L_{\beta \gamma}^{\alpha}=G_{\beta \gamma}^{\alpha}-\Gamma_{\beta \gamma}^{\alpha}
\end{equation}
is the Landsberg (Finsler) tensor. The tensor $L_{\alpha \beta \gamma}(x,v)=g_{\alpha \mu}(x,v) L^{\mu}_{\beta \gamma}(x,v)$ is  symmetric and  $L_{\alpha \beta \gamma}(x,v) v^\gamma=0$.

In general, whenever $\nabla^H$ is used, it should be clear or made clear at which point of $TM$  is the expression evaluated: one often speaks of {\em support vector}. Observe that if $X\colon E\to VE= E\times_M TM$ is a Finsler field whose components do not depend on $v$, we have at the support vector $u$: $\nabla^H_u X=\tilde D_u X$; while at the support vector $X$, $\nabla^H_u X=D_u X$. For this reason, whenever possible we use directly $D$ or $\tilde D$ in place of $\nabla^H$.


A property of the flipped derivative, which is a consequence of the horizontal compatibility of the  Chern-Rund or Cartan connections with the metric, is
\begin{equation}
\tilde{D}_u g_u(X,Y)= g_u(\tilde{D}_u X,Y)+g_u(X,\tilde{D}_u Y),
\end{equation}
for every vector $u\in T_p M\backslash 0$  and fields $X,Y\colon M\to  TM$. The linearity of the map $X\mapsto \tilde{D}_u X$ implies that $\tilde{D}_u$ extends to one-forms and tensors in the usual way, so the previous equation is simply
\begin{equation} \label{aao}
\tilde{D}_u g_u=0.
\end{equation}
Together with Theor.\ \ref{tiu} the next proposition will allow us to work directly  with $D$ or $\tilde D$, reducing recurse to Finsler connections.
\begin{proposition}
For every $u\colon U \to TM\backslash 0$, $U\subset M$, and $X,Y\colon U\to  TM$
\begin{align}
\frac{1}{2} \,\p_X g_u(u,u)=g_u(u,D_Xu)&=\p_u g_u(u,X)-g_u(\tilde D_u u,X)+g_u(u,[X,u]), \label{doi}\\
\p_X g_u(u,Y)-\p_Yg_u(u,X)&=g_u(Y,D_X u)-g_u(X,D_Yu)+g_u(u,[X,Y]). \label{doj}
\end{align}
\end{proposition}

Observe that the first identity in (\ref{doi}) implies that the non-linear parallel transport preserves the Finsler length of vectors.

\begin{proof}
Let $\nabla^H$ be the Cartan or Chern-Rund horizontal covariant derivative, then at the support vector $u$, using (\ref{kkq}) and (\ref{aao})
\begin{align*}
\frac{1}{2} \,\p_X g_u(u,u)&=\frac{1}{2} \,\nabla^H_X g_u(u,u)=g_{u}(u,\nabla_X^H u)=g_{u}(u,D_X u)\\&=g_u(u,\tilde{D}_u X+[X,u])=\p_ug_u(u,X)-g_u(\tilde{D}_u u, X)+g_u(u,[X,u]),
\end{align*}
which proves (\ref{doi}). Similarly, still at the support vector $u$
\begin{align*}
&\p_X g_u(u,Y)-\p_Yg_u(u,X)=\nabla_X^H g_u(u,Y)-\nabla^H_Yg_u(u,X)=g_u(u, \nabla_X^HY-\nabla^H_Y X)\\
&+g_u(\nabla_X^Hu,Y)-g_u(\nabla^H_Y u,X)= g_u(Y,D_X u)-g_u(X,D_Yu)+g_u(u,[X,Y]),
\end{align*}
where we used $\nabla^H_X Y-\nabla^H_Y X=[X,Y]$, see \cite[Sect.\ 5.3.1]{minguzzi14c}.
\end{proof}


The horizontal-horizontal curvature $R^{HH}$ of any Finsler connection is related to the curvature of the non-linear connections as follows (see e.g.\ \cite[Eq.\ (67)]{minguzzi14c})
\[
R^{HH} {}^\alpha_{\ \beta \mu \nu}(x,v) v^\beta=R^\alpha_{\mu \nu}(x,v) .
\]
Bao, Chern and Shen \cite{bao00,shen01}  work with just the Chern-Rund connection, and use a contracted tensor
\begin{equation} \label{hhn}
R^\alpha_{\ \beta}(x,v):= R^\alpha_{\ \beta \mu}(x,v) v^\mu=R_{\textrm{ChR}}^{HH} {}^\alpha_{\ \, \mu \beta \nu}(x,v) v^\mu v^\nu,
\end{equation}
where the Chern-Rund HH-curvature is
\begin{equation} \label{vll}
R^{HH}_{\textrm{ChR}}{}^\alpha_{\ \, \beta \gamma \delta}=\frac{\delta}{\delta x^\gamma} \,\Gamma^\alpha_{\beta \delta}-\frac{\delta}{\delta x^\delta} \, \Gamma^\alpha_{\beta \gamma}+\Gamma^\alpha_{\mu \gamma} \Gamma^\mu_{\beta \delta}-\Gamma^\alpha_{\mu \delta} \Gamma^\mu_{\beta \gamma}.
\end{equation}
Actually, as the first identity in (\ref{hhn}) clarifies, the  tensor  $R^\alpha_{\ \beta}$ depends only on the curvature of the non-linear connection and not on the full Finsler connection. Sometimes, following them, we shall denote it simply $R_v$, where the index $v$ stresses the dependence on the point on $E$. Again, following their terminology we shall also write $Ric(v)=\textrm{tr} R_v$.

The tensor $R_v$ is related to the non-commutativity of the covariant and flipped derivatives as follows
\begin{theorem} \label{tiu}
Let $u$ and $X$ be fields on $M$, and let $u$ be pregeodesic, namely $D_uu=f u$, for some function $f$. Then
\begin{equation} \label{com}
D_X \tilde{D}_u u-\tilde{D}_u D_X u-D_{[X,u]} u=R_u(X).
\end{equation}
\end{theorem}

\begin{proof}
The definition of $HH$-curvature for the Chern-Rund connection reads
\[
\nabla^H_X\nabla^H_YZ-\nabla^H_Y\nabla^H_X Z-\nabla^H_{[X,Y]}Z=R_{\textrm{ChR}}^{HH}(X,Y) Z,
\]
which at the support vector $u$, and with $Y=Z=u$, gives the desired equality. The pregeodesic condition $\tilde{D}_uu=f u$ is used to write, at the support vector $u$, $\nabla^H_X\nabla^H_uu= \nabla^H_X (f u)=(\p_X f) u+ f \nabla^H_X u=(\p_X f) u+ f D_X u=D_X(f u)=D_X \tilde{D}_u u$.
\end{proof}


As a corollary we obtain (see also  \cite[Lemma 6.1.1]{shen01})
\begin{proposition}
Let $x(t,s)$ be a geodesic variation, namely $x_s:=x(\cdot, s)$ is a geodesic for each $s\in (-\epsilon,\epsilon)$. Then defined $J=\p  /\p s\vert_{s=0}$ we have
\begin{equation} \label{cmo}
{\tilde D_{\dot x}}\,\,{\tilde D_{\dot x}} J+ R_{\dot{x}} (J)=0.
\end{equation}
\end{proposition}

\begin{proof}
Immediate from Eqs.\ (\ref{kkq}) and (\ref{com}) using $[J,\dot{x}]=0$ and $D_{\dot{x}}\dot{x}=0$.
\end{proof}


Observe that the notion of Jacobi field depends solely on the spray and on its induced non-linear connection, not on the Finsler metric nor on a choice of Finsler connection. Observe also that if $J$ satisfies Eq.\ (\ref{cmo}), then for every constants $a,b$, $J+(a+bt)\dot x$ satisfies the same equation, thus this arbitrariness can be used to adjust the $\dot x$ component of the variational field at two distinct points or to adjust  $J$ and its first derivative ${\tilde D_{\dot x}} J$ at the same point.


We shall need the following result

\begin{proposition} \label{xll}
The  curvature $R_v$ of a Finsler Lagrangian $\mathscr{L}$ satisfies for every $X,Y\in T_pM$
\begin{align}
R_v(v)&=0, \label{mki} \\
g_v(v, R_v(X))&=0, \label{mko}\\
g_v(X, R_v(Y))&=g_v(Y, R_v(X)). \label{mkp}
\end{align}
\end{proposition}

\begin{proof}
 The first equation is trivial. Equation (\ref{mko}) follows from the skew-symmetry of the Cartan ${HH}$ curvature on the first two indices (see Eqs.\ (67) and (82) of \cite{minguzzi14c}), namely \[R^{HH}_{\textrm{Car}}\,{}_{\alpha \beta \gamma \delta}=-R^{HH}_{\textrm{Car}}\,{}_{\beta \alpha \gamma \delta}.\]
 Equation (\ref{mkp}) follows e.g.\ contracting Eq.\ (85) of  \cite{minguzzi14c}, namely
\begin{equation} \label{don}
R^{HH}_{\textrm{Car}} {}_{\alpha \beta \gamma \delta}-R^{HH}_{\textrm{Car}} {}_{\gamma \delta \alpha \beta}=R^\mu_{\gamma \alpha} C_{\mu \beta \delta }-R^\mu_{\gamma \beta} C_{\mu \delta \alpha}-R^\mu_{\delta \alpha} C_{\mu \beta \gamma}+R^\mu_{\delta \beta} C_{\mu \gamma \alpha}.
\end{equation}
  where $C_{\alpha \beta \gamma}=\frac{1}{2}\frac{\p}{\p v^\gamma} g_{\alpha \beta}$ is the Cartan torsion, with $v^\beta v^\delta$ (in \cite{minguzzi14c} $v$ is denoted $y$). That is, it follows from the symmetry properties of the Cartan ${HH}$ curvature under exchange of the first pair of indices with the second pair of indices. See also \cite[Eq.\ (6.7)]{shen01}.
\end{proof}



By taking as reference pseudo-Riemannian geometry, let us define a Weyl curvature  as follows
\begin{equation} \label{cin}
C^\mu_{\ \gamma}(x,v)=\Big\{R^{HH} {}^{\mu \beta}_{\ \ \, \gamma \delta}-\frac{4}{n-1}\, \delta^{[\mu}_{[\gamma} R^{HH} {}^{\eta \beta]}_{\ \, \ \eta \delta]}+\frac{2 R^{HH} {}^{\eta \xi}_{\ \  \eta \xi}}{n (n-1)}\,  \delta^{[\mu}_{[\gamma} \delta^{\beta]}_{\delta]}\Big\}g_{\beta \nu} v^\nu v^\delta.
\end{equation}
The endomorphism $X\mapsto C_v(X)$  is easily checked to be traceless and such that $C_v(v)=0$. Furthermore, if $v$ is lightlike and $g_v(v,X)=0$ then
\begin{equation}
C_v(X)= T_v(X)+ \textrm{ terms prop.\ to } v, \quad T_v(X):= [R_v- \frac{\textrm{tr} R_v}{n-1} \,\textrm{Id}](X)
\end{equation}
We did not mention which Finsler connection is used to define $C_v$ since, as will be clear in a moment, what really matters is $T_v$ which depends only on the non-linear connection.


\section{The Raychaudhuri equation: timelike case}

Let $H$ be a $C^2$ spacelike hypersurface and let $u$ be a $C^1$ future directed timelike, normalized, geodesic vector field  orthogonal to it (for the notion of normal vector see \cite{minguzzi13c})
\[
D_u u=0, \qquad g_u(u,u)=-1, \qquad \textrm{ker} g_u(u,\cdot)\vert_H=TH.
\]
It is interesting to observe that $H$ has a natural induced metric given by $g_u$. Since $u$ is timelike for $g_u$ we have that $(H,g_u)$ is a Riemannian manifold.

The flow of $u$ propagates the hypersurface $H$ into a foliation $H_s$, $H=H_0$, $u=\frac{d}{d s}$, at least in a neighborhood of $H_0$ before the development of focusing points. Similarly, the flow propagates a vector field $X$ tangent to $S_0$ into a vector field, denoted in the same way, tangent to the foliation and such that $[u,X]=0$. As a consequence, the foliation remains orthogonal to $u$ because, for every vector field $X$ tangent to the foliation,  we have $\p_u g_u(u,X)=0$ as it follows immediately from Eq.\ (\ref{doi}).

In the domain  of $u$ we consider the  vector bundle $V$ which consists of vectors ${X}\in TM$ orthogonal to $u$: $g_u(u,{X})=0$. Clearly, this bundle has $n$-dimensional fibers.
We introduce  a positive definite (space) metric on $V$
\[
h({X},{Y}):=g_u({X},{Y}),
\]
an endomorphism (shape operator, Weingarten map)
\[
b\colon V_p \to V_p,\qquad   {X}\mapsto b({X}):= D_{{X}} u , 
\]
a second endomorphism
 \[
 \bar{R}\colon V_p \to V_p, \qquad \bar{R}({X}):=R_u( X),
 \]
  and a third endomorphism which is the trace free part of $\bar{R}$
 \[
  \bar{Q}\colon V_p \to V_p,\qquad \bar{Q}({X}):=\bar{R}-\frac{1}{n} \,\textrm{tr} \bar{R} \,Id .
 \]
The definition of $b$ is well posed because, due to Eq.\ (\ref{doi}), $g_u(u,D_X u)=0$, while $\bar{R}$ is well posed thanks to (\ref{mko}). The endomorphism $\bar Q$ can be related with the tensor $C$ defined in the previous section, however this relationship will not be used in what follows (by the way $\bar Q$ does not depend on the chosen Finsler connection while $C$ does).




The endomorphisms $b, \bar{R}, \bar{Q}$ are all self-adjoint with respect to $h$. The self-adjointness  of $\bar{R}$ (and hence of $\bar{Q}$) follows from (\ref{mkp}). In order to show that $b$ is self-adjoint, let $X,Y \in V_p$ and let extend them to two commuting vector fields tangent to the foliation and  denoted in the same way. We have using Eq.\ (\ref{doj})
\begin{align*}
h({X},b({Y}))&=g_u(X,D_Y u)=g_u(Y,D_X u)=h({Y},b({X})).
\end{align*}
Let us prove
\begin{align}
\textrm{tr} \bar{R}&=\textrm{tr} R_u=Ric(u), \label{fid}
\end{align}
Indeed let $\{u,e_1,  $ $\cdots,  e_{n}\}$ be a $g_u$-orthonormal basis of $T_pM$  such that  $\{e_i\}$ is a basis of $V_p$. Observe that $h( e_i,  e_j)=g_u(e_i,e_j)=\delta_{ij}$. Thus using (\ref{mki})
\begin{align*}
\textrm{tr} \bar{R}&=\sum_i h( e_i,\bar{R}( e_i))=\sum_i g_u(e_i, R_u(e_i))= \sum_i g_u(e_i, R_u(e_i))-g_u(u,R_u(u))\\&=\textrm{tr} {R}_u.
\end{align*}

Both endomorphisms $\bar{R}$ and $\bar{Q}$ depend on $u$ at the considered point $p$ but not on the whole  geodesic congruence.
\begin{definition}
In  analogy with Lorentzian geometry define the properties:
\begin{itemize}
\item {\em Timelike convergence condition}:  at every $p$ and for every f.d.-timelike vector $u$, $Ric(u)\ge 0$,
\item {\em Timelike genericity condition}:  every inextendible complete f.d.-timelike geodesic admits some point $p$ at which the tangent vector $u$ satisfies  $\bar{R}(p,u)\ne 0$.
\end{itemize}
\end{definition}
This last condition can also be written in terms of the contracted non-linear curvature $R_u$ (here $u_\nu$ is the one-form $g_u(u,\cdot)$)
\[
u_{[\alpha} R_{u\, \beta] [\mu}\,  u_{\nu]}\ne 0, \textrm{ or equivalently } R_u \ne 0.
\]
The first formulation allows one to joint the timelike genericity condition and the null genericity condition (see below) into a {\em causal} genericity condition.



Let $P\colon TM\to V$ be the  projection with kernel $\textrm{Span} (u)$.
The  derivative ${\tilde D_{u}}$, induces a  derivative $X':=P( {\tilde D_{u}} X)$ on sections of $V$, and hence, as usual, a derivative on endomorphisms as follows $E'(X):=(E(X))'-E(X')$. It  can be observed that if $[u,X]=0$ then $X'={\tilde D_{u}} X$ indeed due to Eq.\ (\ref{kkq}) and (\ref{doi}) $g_u(u,{\tilde D_{u}} X)=g_u(u,{ D_{X}} u)=\p_X g_u(u,u)/2=0$.


\begin{proposition}
The Weingarten map satisfies the Riccati equation
\begin{equation} \label{nkv}
b'=-\bar{R}-b^2,
\end{equation}
\end{proposition}


\begin{proof}
Let $X\in V_p$ and extend it in a neighborhood of $p$ so as to remain tangent to the foliation and   in such a way that $[X,u]=0$. Using Eqs.\ (\ref{kkq}) and (\ref{com})
\[
R_u(X)=-\tilde D_u D_X u=-\tilde D_u\tilde D_u X.
\]
Thus
\begin{align*}
b'( X)&=P({\tilde D_{u}} b(X))-b({\tilde D_{u}} X)=P({\tilde D_{u}} {D_X u})-b({{ D_{X}} u})\\ &= P(\tilde D_{u} \tilde D_{u} X)-b(b({X}))=-R_u(X)-b^2({X}),
\end{align*}
which concludes the proof.
\end{proof}
It can also be observed that $h'=0$ because for $X,Y$ vector fields orthogonal to $u$, we have using Eq.\ (\ref{aao})
\begin{align*}
h'(X,Y):=&\, (h(X,Y))'-h(X',Y)-h(X,Y')\\=&\, \p_u g_u(X,Y)-g_u(\tilde{D}_u X, Y)-g_u(X,\tilde{D}_u Y)=0.
\end{align*}
Let us define
\begin{align*}
\theta:&=\textrm{tr} \,b, \\
\bar\sigma:&=b-\frac{1}{n}\, \theta\, {Id},
\end{align*}
so that $\bar{\sigma}$ is the trace-free part of $b$.  They are called {\em expansion} and {\em shear}, respectively.
Let us denote for short $\sigma^2:=\textrm{tr} \bar{\sigma}^2$. A trivial consequence of this definition is $\sigma^2\ge 0$ with equality if and only if $\bar{\sigma}=0$.

Taking the trace and the trace-free parts of  (\ref{nkv})  we obtain
\begin{align}
 \theta'&=-\textrm{Ric}(n)-\sigma^2-\frac{1}{n} \, \theta^2, \qquad (\textrm{Raychaudhuri}) \label{ray} \\
 \bar{\sigma}'&= -\bar{Q}-(\bar{\sigma}^2-\frac{1}{n} \, \textrm{tr} \bar{\sigma}^2 \, {Id})-\frac{2}{n}\, \theta\, \bar{\sigma}, \label{she}
\end{align}
the   term in parenthesis is the trace-free part of $\bar{\sigma}^2$.

%

Let $H$ be a $C^2$ spacelike hypersurface and let $u$ be the $C^1$ future directed timelike  normal.  We  say that the congruence is {\em converging} if $\theta:=\textrm{tr}( {X}\mapsto D_{{X}} u)$ is negative, {\em diverging} if it is positive. This section will be useful in the generalization of Hawking's (1967) singularity theorem or of Hawking and Penrose's singularity theorem (1970).

\section{The Raychaudhuri equation: null case} \label{ral}

In Lorentzian geometry the next result is well known \cite{kupeli87,galloway00}. In the `only if' direction the achronality property follows from the existence of convex neighborhoods. Since the Finslerian theory admits convex neighborhoods this theorem holds as well in the Finslerian case with no alteration in the proof.

\begin{theorem} \label{ppp}
Every $C^2$  hypersurface $H$ is null if and only if it is locally achronal and ruled by null geodesics.
\end{theorem}

Let $H$ be a $C^2$ null hypersurface and let $n$ be a $C^1$ lightlike vector field  tangent to $H$ so that its integral curves are lightlike pregeodesics running over $H$. They satisfy
\[
{\tilde D_{n}} n=D_n n=\kappa n
\]
where $\kappa$ is a function over $H$. The tangent space at a point $p\in H$ is $T_pH=\ker g_n(n,\cdot)$.

\begin{remark} \label{nid}
 Any $C^2$  null hypersurface $H$ (with boundary) can be enlarged to a hypersurface $H'$ by extending  the null generators in the future direction. The question is whether it remains null. The answer is affirmative provided it remains $C^2$.
 Indeed, let $X\in T_pH'$,  at some extended point $p\in H'$. There is a geodesic variation made of hypersurface  generators whose  Jacobi field $J$, $[J,n]=0$, is such that $J(p)=X$. Then using (\ref{kkq}) and (\ref{doi}),
 \[
 \tilde{D}_n g_n(n,J)-\kappa g_n(n,J)=g_n(n,\tilde{D}_n J)=g_n(n,D_Jn)=D_J g_n(n,n)/2=0,
 \]
  which shows that the null condition $g_n(n,J)=0$ will be propagated from $H$ to $H'$.
\end{remark}

On the $C^2$ null hypersurface we consider the  vector bundle $V=TH/\!\!\sim$ obtained regarding as equivalent any two vectors $X, Y\in T_pH$ such that $Y-X\propto n$. Clearly, this bundle has $n-1$-dimensional fibers. Let us denote with an overline $\bar{X}$ the equivalence class of $\sim$ containing $X$. At each $p\in H$, we introduce  a positive definite metric $h(\bar{X},\bar{Y}):=g_n(X,Y)$,
an endomorphism (shape operator, null Weingarten map)
\[
b\colon V_p \to V_p,\qquad   \bar{X}\mapsto b(\bar{X}):= D_{\bar{X}} n:=\overline{D_X n},
\]
 a second endomorphism
 \[
 \bar{R}\colon V_p \to V_p, \qquad \bar{R}(\bar{X}):=\overline{R_n(X)},
 \]
  and a third endomorphism which is the trace free part of $\bar{R}$
 \[
  \bar{C}\colon V_p \to V_p,\qquad \bar{C}(\bar{X}):=\bar{R}-\frac{1}{n-1} \,\textrm{tr} \bar{R} \,Id .
 \]

The definition of $b$ is well posed because, as $D_X$ is linear in $X$ and $D_n n\propto n$ we have $D_{X+a n} n=D_X n+k n$, for some $k$. Moreover, extending $X$ to a vector field tangent to $H$ which commutes with $n$, $g_n(n,D_X n)=g_n(n,\tilde{D}_n X)=\tilde{D}_n g_n(n,X)-g_n(\tilde{D}_nn,X)=0$ which means that $D_Xn\in TH$. The definition of $\bar{R}$ is well posed since $R_n(n)=0$ and by Eq.\ (\ref{mko}) $g_n(n,R_n(X))=0$ which implies that for every $X\in T_pM$, $R_n(X) \in T_pH$.

The endomorphisms $b, \bar{R}, \bar{C}$ are all self-adjoint with respect to $h$. In order to show that $b$ is self-adjoint, let $X,Y \in T_pH$ and let extend them to two commuting vector fields tangent to $H$ and denoted in the same way. We have using Eq.\ (\ref{doj}) with $u=n$
\[
h(\bar{X},b(\bar{Y}))=g_n(X,D_Y n)=g_n(Y,D_X n)=h(\bar{Y},b(\bar{X})).
\]
The self-adjointness of $\bar{R}$ follows from Eq.\ (\ref{mkp}), while that of $\bar{C}$  follows  from that of $\bar{R}$.
Let us prove
\begin{align}
\textrm{tr} \bar{R}&=\textrm{tr} R_n=Ric(n). \label{sid}
\end{align}
Let $\{n,m,e_1,  $ $\cdots,  e_{n-1}\}$ be a $g_n$-basis of $T_pM$  such that $m$ is $g_n$-lightlike, $g_n(m,n)=-1$, and  $\{e_i\}$ is a basis of the spacelike codimension 2 subspace $g_n$-orthogonal to both $n$ and $m$. Observe that $h(\bar e_i, \bar e_j)=g_n(e_i,e_j)=\delta_{ij}$. Thus using (\ref{mki}) and (\ref{mko})
\begin{align*}
\textrm{tr} \bar{R}&=\sum_i h(\bar e_i,\bar{R}(\bar e_i))=\sum_i g_n(e_i, R_n(e_i))\\
&= \sum_i g_n(e_i, R_n(e_i))-g_n(m,R_n(n))-g_n(n, R_n(m))=\textrm{tr} {R}_n.
\end{align*}

\begin{proposition}
We have the equality $\bar{C}(\bar X)=\overline{C_n(X)}$.
\end{proposition}

\begin{proof}
Indeed,
\begin{align*}
\overline{C_n(X)}&=\overline{T_n(X)}=\overline{R_n(X)}- \frac{1}{n-1}(\textrm{tr} \bar{R}) \bar X=[\bar{R}-\frac{1}{n-1} \,\textrm{tr} \bar{R} \,Id](\bar X)=\bar{C}(\bar X).
\end{align*}

\end{proof}
Both endomorphisms $\bar{R}$ and $\bar{C}$ depend on $n$ at the considered point $p$ but not on the whole  geodesic congruence tangent to $H$.
\begin{definition}
In analogy with Lorentzian geometry define the properties:
\begin{itemize}
\item {\em Null convergence condition}: for every f.d.-lightlike vector $n$, $Ric(n)\ge 0$,
\item {\em Null genericity condition}:  every inextendible complete f.d.-lightlike geodesic admits some point $p$ at which the tangent vector $n$ satisfies  $\bar{R}(p,n)\ne 0$.
\end{itemize}
\end{definition}
This last condition can also be written in terms of the contracted non-linear curvature $R_n$
\[
n_{[\alpha} R_{n\, \beta] [\mu}\,  n_{\nu]}\ne 0 .
\]
The  derivative ${\tilde D_{n}}$,  induces a  derivative $\bar{X}':=\overline{{\tilde D_{n}} X}$ on sections of $V$, and hence, as usual, a derivative on endomorphisms as follows $E'(\bar{X}):=(E(\bar{X}))'-E(\bar{X}')$.

\begin{proposition} \label{buu}
Along a generator of $H$ the null Weingarten map satisfies the Riccati equation
\begin{equation} \label{nkw}
b'=-\bar{R}-b^2+ \kappa\, b,
\end{equation}
\end{proposition}

The proof if a Finslerian modification of that given in  \cite{galloway00}.

\begin{proof}
Let $X\in T_pH$ and extend it in a neighborhood of $p$ so as to remain tangent to $H$ and in such a way that $[X,n]=0$. Using Eqs.\ (\ref{kkq}) and (\ref{com}) with $u=n$
\[
R_n(X)=D_X \tilde{D}_nn-\tilde D_n D_X n=\kappa D_X n+(\p_X\kappa)n-\tilde D_n\tilde D_nX.
\]
Thus
\begin{align*}
b'(\bar X)&=\overline{\tilde D_{n} b(\bar X)}-b(\overline{{\tilde D_{n}} X})=\overline{\tilde D_{n} \overline{D_X n}}-b(\overline{{ D_{X}} n})= \overline{\tilde D_{n} D_X n}-b(b(\bar{X}))\\
&=\overline{\tilde D_{n} \tilde D_n X}-b^2(\bar{X})=-\overline{R_n(X)}+\kappa \overline{D_X n}-b^2(\bar{X}),
\end{align*}
which concludes the proof.
\end{proof}
It can also be observed that $h'=0$; indeed
\begin{align}
h'(\bar{X},\bar{Y}):=&(h(\bar{X},\bar{Y}))'-h(\bar{X}',\bar{Y})-h(\bar{X},\bar{Y}')\\
=&\tilde{D}_n g_n(X,Y)-g_n(\tilde{D}_n X,Y)-g_n(X,\tilde{D}_n Y)=0.
\end{align}
Let us define the {\em expansion} $\theta:=\textrm{tr} \,b$ and the  {\em shear}
\begin{align*}
\bar\sigma:=b-\frac{1}{n-1}\, \theta\, {Id},
\end{align*}
so that $\bar{\sigma}$ is the trace-free part of $b$.
Let us denote for short $\sigma^2:=\textrm{tr} \bar{\sigma}^2$. A trivial consequence of this definition is $\sigma^2\ge 0$ with equality if and only if $\bar{\sigma}=0$.

Taking the trace and the trace-free parts of  (\ref{nkw})  we obtain
\begin{align}
 \theta'&=-\textrm{Ric}(n)-\sigma^2-\frac{1}{n-1} \, \theta^2+\kappa\, \theta, \qquad (\textrm{Raychaudhuri}) \label{raw} \\
 \bar{\sigma}'&= -\bar{C}-(\bar{\sigma}^2-\frac{1}{n-1} \, \textrm{tr} \bar{\sigma}^2 \, {Id})-\frac{2}{n-1}\, \theta\, \bar{\sigma}+\kappa \,\bar{\sigma}. \label{shw}
\end{align}
The   term in parenthesis is the trace-free part of $\bar{\sigma}^2$ and vanishes in the physical four dimensional spacetime case ($n=3$).

Let $S$ be a codimension 2, oriented,  $C^2$ spacelike manifold. Let $p\in S$, since $T_p S$ does not intersect the future causal cone (saved for the origin), by the convexity of this cone there are exactly two hyperplanes $ B_p^{\pm}\subset T_pM$ containing $T_pS$ and tangent to the cone. These hyperplanes  determine two future lightlike vectors  $n^{\pm}$ up to a proportionality constant (see \cite{minguzzi13c}): $B_p^{\pm}=\ker g_{n^\pm}(n^\pm, \cdot)$. Let us denote in the same way  a $C^1$ choice of vector field $n^\pm$ over $S$, which exists by orientability.
 Its exponential map generates, at least locally, a $C^2$ locally achronal null hypersurface $H^{\pm}$. The expansion is defined as above $\theta^\pm=\textrm{tr}(X\mapsto \overline{D_X n^\pm})$.

\begin{definition}
The manifold $S$ is {\em trapped } if  $\theta^+,\theta^- < 0$, {\em weakly trapped} if $\theta^+,\theta^- \le 0$, {\em marginally trapped} if weakly trapped and either $\theta^+$ or $\theta^-$ is negative.
\end{definition}

This section will be useful in the generalization of Penrose's singularity theorem (1965), of Gannon's (1975), and of many others.

%
%
%

\section{Conjugate/focusing points and maximization} \label{hpd}

The next proposition is stated so as to hold in both the null and timelike cases.
The proof of the chronality statement (a) is  more topological and in the end much simpler than that given in textbooks.

The traditional textbook proofs \cite[Prop.\ 4.5.12]{hawking73} \cite[Prop.\ 10.48]{oneill83} \cite[Lemma 4.6.15]{kriele99}  for the null case of (a) are somewhat incomplete since the compactness argument used does not prove that the variation is causal near the endpoints. The proof given by Beem et al.\ seems to be fine in this respect \cite[Theor.\ 10.72]{beem96}. A cleaner proof can also pass through the more general Causality Lemma given by Galloway in \cite{galloway96b}.
Here this lemma is   included in  (b) which also shows that the weak null convexity condition in \cite{galloway96b} is not required. The timelike case is shown to be a corollary of the null case using a product trick. We recall that the Lorentz-Finsler length of a causal curve $\sigma$ is $\int_\sigma \sqrt{-g_{\dot \sigma}(\dot \sigma,\dot \sigma)}\dd t$.

\begin{proposition}
Let us consider a null (resp.\ timelike and normalized) geodesic  congruence orthogonal to a codimension one hypersurface $H$. Let $\gamma$ be a half-geodesic belonging to the congruence,    $p=\gamma(0)\in H$. The first focusing point $q=\gamma(t_q)$, $t_q>0$, of the congruence on $\gamma$ is the first point where $\theta \to -\infty$.

Assume that this focusing point exists.
\begin{itemize}
\item[(a)] Let $U_p$ be a compact neighborhood of $p$ and let $r= \gamma(t_r)$, $t_r>t_q$, then there is a timelike curve from $U_p\cap H$ to $r$ whose Lorentzian length is larger than the Lorentz-Finsler length $l(\gamma\vert_{[0,t_r]})$ of the geodesic between $p$ and $r$.
\item[(b)] Let $U_q$ be a compact neighborhood of $q$ and $H'$ a hypersurface orthogonal to $\gamma$ at $q$. Then there is a timelike curve from $U_p\cap H$ to $U_q\cap H'$ with Lorentzian length     larger than the Lorentz-Finsler length $l(\gamma\vert_{[0,t_q]})$ of the geodesic between $p$ and $q$.
\end{itemize}
Concerning the existence of focusing points, let us assume that the null (resp.\ timelike) convergence condition holds.  We have:
\begin{itemize}
\item[(c)] If $\theta(p)<0$ then the geodesic congruence necessarily develops a focusing point $q=\gamma(t_q)$, provided $\gamma$ extends at sufficiently large affine parameters.
\item[(d)] If there is a compact manifold $T\subset H$ transverse to the congruence and such that $\theta<0$ on it then every geodesic crossing $T$ develops a focusing point in the future within a certain bounded affine parameter provided the affine parameters extend sufficiently far.
\end{itemize}
\end{proposition}

Observe that $H$ is transverse to the congruence in the timelike  case, and contains the congruence in a neighborhood of $p$ in the null case (it is a null hypersurface). The statement does not claim that the focusing point is inside $H$, for we assume that in $H$ the congruence is determined by a $C^1$ vector field (thus $H$  has edge in the null case).

\begin{remark}
One might ask whether the theorem applies to congruences of half-geodesics issued by a single point $p$. The answer is affirmative. In the null case it is sufficient to consider a convex neighborhood $C$ of $p$ and take as $H$ the exponential map (on $C$) of the future light cone at $p$ minus the zero vector. In the timelike case it is sufficient to recall Gauss' lemma, and chosen a convex neighborhood $C$ of $p$  define $H$ as the exponential map (on $C$) of the subset of $T_pM$ which consists of future directed timelike unit vectors. Then the conjugate point $q$ for $p$ becomes a focusing point for $H$ and the theorem applies giving the usual results for conjugate points.
\end{remark}

\begin{proof}
In the null case we denote with $n$ the geodesic field tangent to the congruence, $D_nn=0$, while we use $u$ in the timelike normalized case.
Let us first prove that at the first focusing point $q=\gamma(t_q)$ we have $\theta\to -\infty$. Indeed, let $\{e_i\}$ be a basis at $V_p$  and let us transport it over $\gamma$ through the condition $e'_i=0$. The Jacobi equation provides a linear map from $V_p$ to $V_q$ whose Jacobian is $J^j_i(t_q)$ where $J_i:=J^j_i(t) e_j$ is the value at $\gamma(t)$ of the Jacobi field  $J_i$ whose initial condition is $J_i(0)=e_i$, $J_i'(0)=D_{e_i} n$ (in the timelike case replace $n$ with $u$ in this and the next formulas). Observe that $J_i$ is such that $[J_i, n]=0$ at $t=0$ thus at every later instants because by  Eqs.\ (\ref{kkq}) and (\ref{com}) we have the linear differential equation in $[J_i,n]$, $\tilde{D}_n [J_i,n]+D_{[J_i,n]} u=0$.

Observe also that $J^j_i(0)=\delta^j_i$; thus by continuity $\det J^j_i > 0$ at least before the first focusing point.
As $q$ is a focusing point this linear map is not injective, that is $\det J^j_i(t_q)=0$,
 which implies $\ln \det J^j_i(t) \to -\infty$ for $t\to t_q$, and hence $\frac{\dd }{\dd t} \ln \det J^j_i(t_k)\to -\infty$ for some sequence $t_k\to t_q$. Using Jacobi's formula for the derivative of a determinant we get
\begin{align*}
\frac{\dd }{\dd t} \ln \det J^j_i(t_k)=&\, (J^{-1})^i_j \frac{\dd }{\dd t} J^j_i(t_k)= \textrm{tr}  (J\mapsto \! \tilde D_n J)(t_k)= \textrm{tr}  (J\mapsto  D_J n)(t_k)= \theta(t_k),
\end{align*}
thus $\theta\to -\infty$ at the first focusing point.

Conversely, let us consider the first point $q$ where $\theta\to -\infty$.
Let $T$ be a $k$-dimensional ($k=n$ in the timelike case and $k=n-1$ in the null case) manifold $T\subset H$, $p\in T$, transverse to the congruence  and determined by the restriction of the vector field ($u$ in the timelike case, $n$ in the null case) on $T$. There must be a focusing point on $\gamma$ at $q$ or before $q$ for otherwise the exponential map from $T$ would provide a well defined local diffeomorphism in a neighborhood of every vector $t \dot{\gamma}(0)$, $t\in [0,b]$, $\exp_p(b \dot{\gamma}(0))=q$,  which would imply that the vector field   ($u$ in the timelike case, $n$ in the null case) is well defined and $C^1$ in a neighborhood of $\gamma$, and so its divergence $\theta$ would be well defined and finite at $q$, a contradiction.  Thus the focusing point can only be $q$ since $\theta\to -\infty$ at the focusing points $q$ would not be the first point where the expansion diverges.

%
%

Let us prove (c). Let us consider the Raychaudhuri equation evaluated on the half-geodesic $t\to \gamma(t)$ with initial point $p=\gamma(0)$.
The Raychaudhuri equation gives the inequality $\frac{d \theta}{d t}\le -\frac{1}{k} \,\theta^2$. As $\theta(t) \le k/[t-k/(-\theta(p))]$ we have $\theta \to -\infty$ at  some $0<t_q<k/(-\theta(p)$ provided the affine parameter extends sufficiently far.

Statement (d) is is clear given the compactness of $T$.

Let us prove (a) in the null case. Let $C$ be a convex neighborhood of $q$, and let us consider three points in sequence  $b<_C q<_C r$ over $\gamma\cap C$. The geodesic congruence at $b$ forms a $C^2$ hypersurface $H$ in a neighborhood of $b$ because there is no focusing point in the segment $\gamma\vert_{[0,t_b]}$. Similarly, the exponential map of the past light cone at $r$ provides a $C^2$ hypersurface $\Sigma$ containing the segment  $\gamma\vert_{[t_b,t_q]}$. By construction they are tangent but when regarded as local graphs they might have different second order Taylor expansion at $b$. Below we are going to compare the second derivatives in some directions.

Observe that the second derivative of the graphing function of $\Sigma$ is  bounded on a compact neighborhood of $q$, however for what concerns $H$,
$\theta(t) \to -\infty$ for $t\to t_q$, which proves that the Weingarten map of $H$ is not bounded there. Since $\theta=tr(X\to \overline{D_X n})$, there is some large negative eigenvalue of this map, namely a $h$-normalized vector $e^{(k)}$  such that $D_{e^{(k)}} n=\lambda_k e^{(k)}+s_k n$, $\lambda_k \to -\infty$ as $t_k\to t_q$, thus $g_n(e^{(k)}, D_{e^{(k)}} n)=h(e^{(k)}, \overline{D_{e^{(k)}} n})=\lambda_k\to -\infty$. Let us extend each $e^{(k)}$ in a neighborhood of $\gamma(t_k)$ so as to remain in $TH$. Using the Cartan or Chern-Rund connection we have at the support vector $n$, $g_n(n,\nabla^H_{e^{(k)}}e^{(k)})\to +\infty$ (in order to grasp the meaning of this equation notice that one can choose local coordinates so that $\Gamma^\alpha_{\beta \gamma}(\gamma,\dot \gamma)=0$ in a neighborhood of $t_q$). This equation clarifies that $H$, at least in some directions, bends so much below the exponential map of the tangent plane $\ker g_n(n,\cdot)(b)$  that taking $b$ sufficiently close to $q$ it bends more than $\Sigma$ and hence enters the chronological past of $q$. As a consequence $q$ belongs to the chronological future of $H\cap U_p$.

In the null case the statement (b) involving $H'$ is similar, we just need to use  the past lightlike congruence issued from $H'$ and containing $\gamma$ in place of the exponential map of the past light cone at $r$. Observe that $H'$ has bounded second fundamental form near $q$ and is tangent to $H$ at $b$. From here the argument is exactly the same as before.

The timelike cases for (a) and (b) are  in fact  corollaries of the null case. It is sufficient to apply the null case to the spacetime $M^\times=M\times \mathbb{R}$ endowed with a Finsler Lagrangian which on causal vectors takes the expression $\mathscr{L}^\times((x,y),(v,w))=\mathscr{L}(x,v)+w^2/2$ (it exists by the results of \cite{minguzzi14h}), and lift the timelike normalized congruence to a lightlike congruence as follows. Given the set $H$ on $M$, and the timelike geodesics $\gamma(s)$ starting  from $H$  we consider the set $H \times \{0\}$, and the lightlike geodesics starting from it (light lifts) $(\gamma(t),t)$. We denote with $\tilde{H}$ the hypersurface spanned by these geodesics in a neighborhood of $H\times \{0\}$. The vector field $n=(u,1)$ is tangent to $\tilde{H}$ where $u$ is the normalized timelike field on $M$ which generates the timelike congruence. The vector field $n$ is also normal to $H\times \{0\}$, thus $\tilde{H}$ is a $C^2$ null hypersurface which is the local exponential map of a lightlike normal bundle to $\tilde{H}$. It can be further enlarged extending the generators as long as it remains $C^2$ (Remark \ref{nid}). Since $k=(0,1)$ generates an isometry of $(M^\times,\mathscr{L}^\times)$, ($\mathscr{L}^\times$ is independent of $y$) it is Killing \cite{knebelman29}. As a consequence, over the light lift geodesics $g_n(n,k)$ is constant, thus $k$ cannot become tangent to $\tilde H$ since it is transverse to it over $H\times \{0\}$. In conclusion, $\tilde H$ remains transverse to $k$ before the development of focusing points.
As a consequence, the constructed geodesic  congruences on $M$ and $M^\times$ are such that the expansions at $\gamma(t)$ and $(\gamma(t),t)$ coincide, and the projection establishes a correspondence between first focusing points.

The light lift can be defined for any timelike curve starting at $H$, by imposing the extra-coordinate to be the proper time of the lifted curve. The result (a) for the null case implies that  there is a timelike curve $(\sigma(\tau), f(\tau))$ from $(\gamma(0),0)$ to $(\gamma(t_r),t_r)$, $t_r>t_q$, whose projection $\sigma$ is parametrized with respect to proper time. This timelike condition reads $1=-g_{\dot{\sigma}}(\dot\sigma,\dot\sigma)(t)>(\dot f)^2$, which taking the square root, integrating, and using $\int \vert \dot f\vert \dd \tau\ge \int  \dot f \dd \tau=t_r$,  gives that the Lorentz-Finsler length of $\sigma$ is larger than the Lorentz-Finsler  length of $\gamma\vert_{[0,t_r]}$, which proves (a) for the timelike case. The proof of (b) is similar, one has to consider the set $H'\times \{t_q\}$ and the local null hypersurface $\tilde{H}'$ orthogonal to it and tangent to the light lift of $\gamma$.

 It can be observed that the timelike convergence condition for $M$ coincides with the null convergence condition for $M^\times$ thus (c) and (d) in the timelike case could also be regarded as corollaries of the null case.
\end{proof}

\section{Some useful generalizations}

The notions of convex neighborhoods, continuous causal curve, and the basic elements of casuality theory not involving curvature have been first studied in \cite{minguzzi13d} where it has been shown that they translate word for word from the Lorentzian domain. For instance, it has been shown that Finsler spacetimes admit convex neighborhoods, that a kind of reverse Cauchy-Schwarz inequality holds \cite{minguzzi13c} and that causal geodesics locally maximize the Lorentz-Finsler length \cite[Theor.\ 6]{minguzzi13d}. Since every causal curve can be covered by convex neighborhoods, by the usual interpolation arguments \cite[Prop.\ 2.8]{lerner72}, every causal curve is either an achronal lightlike geodesic or its endpoints are connected by a timelike deformation of the curve.
In this section we wish to make a few steps towards more complex generalizations.

According to Hawking and Ellis \cite{hawking73} a future directed continuous causal curve $\gamma\colon [a, b] \to M$, is a continuous curve such that for every
open convex normal set $C$ intersecting $\gamma$, whenever $\gamma([t_1, t_2]) \subset C$, $t_1 < t_2$, the points $\gamma(t_1)$ and $\gamma(t_2)$ are connected by a future directed causal geodesic contained in $C$. This definition can be imported word for word to the realm of Finsler spacetimes. It has been proved in \cite{minguzzi13d} that these curves coincide with those Lipschitz curves which once parametrized with respect to an auxiliary Riemannian metric are future directed causal almost everywhere.

The notion of  continuous causal curve is particularly convenient because it makes reference to the local causal order in a convex neighborhood and not to the metric. As a consequence,   the results on limit curve theorems given in \cite{beem96,minguzzi07c}, being based on this notion, generalize word for word both in their statements and in their proofs since the only technical premise used there is the existence of convex neighborhoods (see also \cite{minguzzi14h}).

\begin{proposition}
All the results on limit curve theorems in \cite{beem96,minguzzi07c} generalize to the Lorentz-Finsler case (and to $C^{1,1}$ metrics).
\end{proposition}

Among those it is worth to recall a few results.
As in Lorentzian geometry the Lorentz-Finsler length of a continuous causal curve is defined as the greatest lower bound of the lengths of the interpolating causal geodesics. Because of the local Lipschitz condition, this length can be calculated with the usual integral
\begin{equation}
l(\gamma)=\int_\gamma \sqrt{-g_{\dot{x}}(\dot x,\dot x) } \,\dd t.
\end{equation}
When discussing limits of curves it is convenient to parametrize causal curves with the $h$-length of an auxiliary complete Riemannian metric.

\begin{lemma}
A continuous causal curve once parametrized with respect to h-length has a domain unbounded from above iff future inextendible and unbounded from below iff past inextendible.
\end{lemma}

The notion of $h$-uniform convergence refers to uniform convergence in the metric space $(M,d_0)$ induced by the Riemannian metric $h$, and according to $h$-parametrization \cite{minguzzi07c}. On compact subset it is actually independent of the metric $h$ used as any two Riemannian metrics are there Lipschitz equivalent.

\begin{proposition}
If the continuous causal curves $\gamma_n\colon I_n \to M$ parametrized with respect to h-length converge h-uniformly
on compact subsets to $\gamma\colon I \to M$, then $\gamma$ is a continuous causal curve.
\end{proposition}

We have
\begin{proposition}
The length functional is upper semi-continuous with respect to uniform convergence on compact subsets (see \cite[Theor.\ 2.4]{minguzzi07c} for details).
\end{proposition}

The Lorentz-Finsler distance function $d\colon M\times M\to [0,+\infty]$ is defined as usual by
\[
d(p,q)=\sup_\gamma l(\gamma)
\]
where $\gamma$ is the generic $C^1$
causal curve connecting $p$ to $q$. If there is no causal curve connecting $p$ to $q$ then it is understood that $d(p,q)= 0$. The typical limit curve theorem is (parametrization and uniform convergence are those induced by $h$)

\begin{proposition}
If $p$ is an accumulation point for a sequence of inextendible causal curves $\gamma_n\colon I_n \to M$ then there is a subsequence converging uniformly to an inextendible continuous causal curve $\gamma\colon I\to M$ passing through $p$.

Furthermore, if $\gamma_n$ are distance maximizing between any pair of points then the same property holds for $\gamma$ (many improvements, for sequences with endpoints or for limit maximizing sequences actually hold, see \cite[Theor.\ 2.4]{minguzzi07c} for details)
\end{proposition}

\begin{remark}
 Since limit curve theorems are so central for the development of causality theory we can also infer from them a number of other results.
The properties of future and achronal sets or of horismos, the properties of domains of dependence, e.g. the strong casuality in the interior, the properties of Cauchy horizons, e.g. the property of being generated by  lightlike geodesics,   the whole causal ladder of spacetimes as improved in \cite{minguzzi07f,minguzzi08b}, including the placement of the non-imprisonment properties, and the transverse ladder,  the possibility of recovering the causal relation from the  family of time functions, all generalize trivially.
\end{remark}

Many properties of the Lorentz-Finsler distance function can be generalized. It is worth to mention (the proof is as in \cite[Lemma 4.4]{beem96})
\begin{proposition}
The Lorentz-Finsler distance is lower semi-continuous.
\end{proposition}
 The proof of the next result is also unaltered \cite[Lemma 4.5]{beem96}.
\begin{proposition}
In a globally hyperbolic spacetime the Lorentz-Finsler distance is finite and continuous.
\end{proposition}
The Avez-Seifert connectedness theorem \cite{hawking73} generalizes with no alteration in proof.
\begin{proposition}
In a globally hyperbolic spacetime any two causally related events $p< q$ are connected by a maximizing causal geodesic.
\end{proposition}

If the dimension of the Finsler spacetime is larger than two the Legendre map  $v \mapsto g_v(v,\cdot)$ is a bijection \cite{minguzzi13c}.
The gradient of a function $f\colon M \to \mathbb{R}$, is the vector field related to $\dd f$ by the Legendre map, thus the unique vector such that $g_{\nabla f}(\nabla f,\cdot)=\dd f$.  Observe that only for reversible metrics $\nabla (-f)=-\nabla f$ and we do not  assume reversibility.

A function which increases over every f.-d.\  casual curve is a {\em time function}. A {\em temporal}  function is a function $f\colon M \to \mathbb{R}$ such that for every $X$ belonging to the future causal cone, $\p_X t>0$. Clearly every temporal function is a time function.

\begin{proposition} \label{aii}
Let $(M,\mathscr{L})$ be a Finsler spacetime with dimension larger than two. Let $t\colon M \to \mathbb{R}$, then $t$ is a temporal function if and only if $\nabla (-t)$ is future directed timelike.
\end{proposition}

Thus a temporal function $t$, is a function such that $\nabla (-t)$ is future directed timelike.

\begin{proof}
Let $p=g_{\nabla (-t)} (\nabla(-t),\cdot)$. For every $X$, $-\dd t(X)=p(X)$; thus by the results of \cite[Sect.\ 2.4-2.5]{minguzzi13c} there is a unique vector $u$ such that $p=g_u(u,\cdot)$ and this vector is future directed timelike. Conversely, if $\nabla (-t)$ is future directed timelike, then by the reverse Cauchy-Schwarz inequality \cite{minguzzi13c}, for every future directed causal vector $X$,
\[
\p_X t=-g_{\nabla (-t)} (\nabla(-t),X)\ge \big(-g_X(X,X)\big)^{1/2} \big[-g_{\nabla (-t)}(\nabla (-t),\nabla (-t))\big]^{1/2}
\]
with equality only if $X$ and $\nabla (-t)$ are proportional which implies $\p_X t>0$ in any case.
\end{proof}

A Finsler spacetime $(M,\mathscr{L})$ is stably causal if there is another Finsler Lagrangian $\mathscr{L}'$ which is causal and whose f.d.-timelike cones contain the f.d.-causal cones of $\mathscr{L}$.
The next result follows from the previous one and Fathi and Siconolfi's  \cite[Theor.\ 1.1]{fathi12}

\begin{proposition} \label{moa}
Let $(M,\mathscr{L})$ be a smooth Finsler spacetime with dimension larger than two.
It is stably causal if and only if it admits a smooth time function $t\colon M\to \mathbb{R}$ such that $\nabla(-t)$ is future directed timelike.
\end{proposition}

\begin{proof}
The only if direction follows as said from \cite[Theor.\ 1.1]{fathi12} and the Prop.\ \ref{aii}. For the other direction we can find easily a Lorentzian metric with wider light cones, just observe that at $p\in M$, $\ker \dd t$ is a hyperplane of $T_pM$ passing through the origin while the Finsler f.d.-causal cone is a convex cone pointed at the origin, thus between the two we can find an elliptic (round) cone which varies smoothly with $p$.
\end{proof}

%
%

The next result uses Fathi and Siconolfi's  \cite[Theor.\ 1.3]{fathi12}.

\begin{proposition}
A smooth Lorentz-Finsler spacetime $(M,\mathscr{L})$ is globally hyperbolic if and only if it admits a smooth Cauchy hypersurface if and only if it admits a smooth onto time function $t\colon M \to \mathbb{R}$ such that the level sets $S_c=\{p: t(p)=c\}$ are smooth spacelike Cauchy hypersurfaces.

In this case $M$ is diffeomorphic to $\mathbb{R} \times S$ where the first coordinate is the function mentioned above and the level sets $S_a$ are  diffeomorphic to $S$, $S$ being the quotient of $M$ under the flow of $\nabla (-t)$. The integral lines $\p_t$ are timelike  and the metric $g_v$ reads for $v$ sufficiently close to $\p_t$
\begin{equation} \label{ldo}
g_{v}= -a_{v}^2(\dd t+b_v)^2+ h_v,
\end{equation}
where $a_v, b_v, h_v$ are respectively: a positive function, 1-form and Riemannian metric over $S_t$ (which can be seen as analogous time-dependent objects on $S$), all positive homogeneous of degree zero in $v$.
Moreover, $\p_t=a_{\p_t}^2 \nabla(-t)$,
\[
-g_{\nabla(-t)}(\nabla(-t),\nabla(-t))=1/a_{\p_t}^2, \qquad b_{\p_t}=0 ,
\]
and $h_{\p_t}(t,s)$ is the metric induced on $S_t$,  $h_{\p_t}(t,s)(\p_t,\cdot)=0$. 
\end{proposition}

Although the previous result provides the typical metric splitting of a globally hyperbolic spacetime $(M, g_{\p_t})$,  it cannot be obtained studying $(M, g_{\p_t})$ since this spacetime is not necessarily globally hyperbolic as the light cones of $g_{\p_t}$ could be wider than those of $(M,\mathscr{L})$.

As another remark, since every Finsler spacetime admits globally hyperbolic neighborhoods (this is immediate from \cite[Lemma 1, Sect.\ 1.4]{minguzzi14c}), its metric can always be written locally as in  Eq.\ (\ref{ldo}).

\begin{proof}
The proof that global hyperbolicity implies the existece of a smooth Cauchy time function is given in \cite{fathi08} for more general cone structures. The fact that the existence of such time function implies global hyperbolicity is immediate from the validity of limit  curve theorems, applying the usual arguments used in Lorentzian geometry.
For the remainder, let us consider the vector field $V=\nabla (-t)$ so that $g_{V}(V,\cdot)=-\dd t$. By the results of \cite{minguzzi13c}, since $\ker \dd t$ is spacelike and $-\dd t$ has negative value on future causal vectors we have that $V$ is future directed timelike, in particular $\p_V t=-g_V(V,V)>0$. The future directed timelike vector field $W=V/[-g_V(V,V)]$ (observe that $g_W(W,W)g_V(V,V)=1$) induces a flow $\varphi_a\colon M\to M$ such that $\varphi_a(S_b)=S_{b+a}$. Each integral line of $W$ intersects every $S_a$, $a\in \mathbb{R}$, which are therefore all diffeomorphic to a quotient manifold $S=M/W$. For every $p\in M$ there is $s\in S_0$ such that $p=\varphi_t(s)$. The map $(t,s) \mapsto p(t,s)$ is the searched diffeomorphism. Observe that by construction $W=\p_t$ and since $-g_W(W,W)>0$ we have for $v$ sufficiently close to $\p_t$, $-g_v(\p_t,\p_t)>0$, from which it follows immediately that the Lorentzian metric $g_v$ can be written as in Eq.\ (\ref{ldo}) for suitable tensors $a_v,b_v, h_v$.
Since the equality $g_{V}(V,\cdot)=-\dd t$ reads $[-g_W(W,W)]^{-1}g_{W}(W,\cdot)=-\dd t$ we get $b_{\p_t}=0$. Thus if $X,Y$ are tangent to $S_t$, $g_{\p_t}(X,Y)=h_{\p_t}(X,Y)$ which proves that $h_{\p_t}$ is the induced metric. The remaining identity follows from $g_{\p_t}(\p_t,\p_t)=-a_{\p_t}^2$.
\end{proof}

As a final result, it is interesting to observe that the local differentiability properties of horizons \cite{beem98,chrusciel98b,chrusciel98,chrusciel01} depend only on their local semi-convexity \cite{minguzzi14d}. The local semi-convexity of the horizons is obtained through a simple argument which uses  properties of the exponential map which are preserved in the Finslerian case. Thus we can conclude

\begin{proposition}
All the local results on the differentiability of horizons extend word for word to the Finslerian domain (e.g.\ horizons are differentiable precisely at points belonging to just one generator; the horizon is $C^1$ on the set of differentiability points;  the set  of non-differentiability points of $H$  is  countably $\mathcal{H}^{n-1}$-rectifiable, thus its Hausdorff dimension is at most $n-1$; etc.).
\end{proposition}

\section{Conformal transformations}

A diffeomorphism $f\colon M \to M'$ is an isometry between pseudo-Finsler spaces $(M, \mathscr{L})$ and $(M', \mathscr{L}')$,  if $\mathscr{L}=f^*\mathscr{L}'$ where $(f^*\mathscr{L}')(x,v)=\mathscr{L}'(f(x),f_* (v))$. By positive homogeneity this condition can be equivalently written  $g=f^*g'$. We have a {\em conformal transformation} if $g$ and $f^* g'$ are proportional at each point of $E$. Actually, as shown by Knebelman \cite{knebelman29b,hashiguchi76} the conformal factor is the pullback of a function on $M$, that is, it is independent of $v$.

Let us investigate how the spray, the non-linear connection, and the Ricci scalar $Ric(v)$ transform under a conformal replacement $\tilde{g}=\Omega^2 g$.
Denoting with $\gamma^\alpha_{\beta \gamma}$ the usual Christoffel symbols, and setting $v_\gamma=g_{v\, \gamma \beta} v^\beta$, $\ln \Omega^{,\alpha}=g_v^{\alpha \delta} \ln \Omega_{,\delta}$
\[
\tilde\gamma_{\beta \gamma}^{\alpha}(x,v)=\gamma_{\beta \gamma}^{\alpha}(x,v)+\delta^\alpha_\beta \ln \Omega_{, \gamma}+\delta^\alpha_\gamma \ln \Omega_{, \beta}-g_{\beta \gamma} \,  \ln \Omega^{, \alpha},
\]
we obtain using  (Eq.\ (\ref{axu})) $2G^\alpha=\gamma^\alpha_{\beta \gamma} v^\beta v^\gamma$,
\begin{equation} \label{syh}
\tilde G^\alpha(x,v)= G^\alpha(x,v)+ v^\alpha \p_v \ln \Omega-\frac{1}{2}\, g_v(v,v) \ln \Omega^{,\alpha} .
\end{equation}
From this equation we obtain the generalization of a well known result
\begin{proposition} \label{sin}
The unparametrized lightlike geodesics are conformally invariant, while the parametrization changes as follows: $\dd \tilde t/\dd t=\Omega^2$.
\end{proposition}
Since the non-linear connection is $N^\alpha_\beta=\p G^\alpha/\p v^\beta$, it transforms as follows
\[
\tilde N_{ \gamma}^{\alpha}(x,v)=N_{ \gamma}^{\alpha}(x,v)+v^\alpha \ln \Omega_{, \gamma}+\delta^\alpha_\gamma \p_v \ln \Omega- v_\gamma \,  \ln \Omega^{, \alpha} +g_v(v,v) C_\gamma^{\alpha \delta} \ln \Omega_{,\delta} .
\]
By Eq.\ (\ref{nsp}) a variation $\delta N^\alpha_\beta=\tilde N^\alpha_\beta -N^\alpha_\beta$ of the non-linear connection induces a variation of the non-linear curvature given by
\[
\delta R^\alpha_{\beta \gamma}=\nabla^{HB}_{\beta} \delta N^\alpha_\gamma -\nabla^{HB}_{\gamma} \delta N^\alpha_\beta -\delta N^\nu_\beta \frac{\p \delta N^\alpha_\gamma}{\p v^\nu}+\delta N^\nu_\gamma \frac{\p \delta N^\alpha_\beta}{\p v^\nu},
\]
where $\nabla^{HB}$ is the Berwald horizontal covariant derivative. From Eq.\ (\ref{hhn}), since  the horizontal derivative of $v^\alpha$ vanishes and $\frac{\p (\delta N^\alpha_\gamma v^\gamma)}{\p v^\nu}=2 \delta N^\alpha_\nu$
\[
\delta Ric(v)=\nabla^{HB}_{\alpha} (\delta N^\alpha_\gamma v^\gamma) - v^\gamma \nabla^{HB}_{\gamma} \delta N^\alpha_\alpha -\delta N^\nu_\alpha \delta N^\alpha_\nu+v^\gamma\delta N^\nu_\gamma \frac{\p \delta N^\alpha_\alpha}{\p v^\nu}.
\]
 In order to simplify this expression for the found  $\delta N$ it is convenient to replace the Berwald derivative with the Chern-Rund derivative,  by making use of  a few facts. The Landsberg tensor can be written $L_{\alpha \beta \gamma}=v^\sigma \nabla^{HB}_\sigma C_{\alpha \beta \gamma}=-\frac{1}{2} \nabla^{HB}_\alpha g_{\beta \gamma}$. Since it is annihilated by $v$, we have   $v^\gamma \nabla_\gamma^{HB}=v^\gamma \nabla_\gamma^{HC}$ where $\nabla^{HC}$ is the horizontal Cartan (Chern-Rund) covariant derivative. After a long but straightforward calculation we arrive at (for shortness we set $\varphi=\ln \Omega$, $\tilde{g}=e^{2\varphi} g$)
  \begin{align}
(\delta R_v)^\alpha_\beta &= v^\alpha \big[   v^\nu \varphi_{,\nu \vert \beta}-\varphi_{,\beta} (\p_v \varphi)-g_v(v,v) C^{\mu\nu}_{\beta} \varphi_{,\mu}\varphi_{,\nu}\big] \nonumber \\
&\quad \ +v_\beta \big[v^\alpha g^{\mu \nu}\varphi_{,\mu} \varphi_{, \nu}+g^{\alpha \delta} v^\gamma \varphi_{,\delta \vert \gamma}  -g_v(v,v) C^{\alpha \mu \nu} \varphi_{, \mu} \varphi_{, \nu} - g^{\alpha \delta} \varphi_{, \delta} (\p_v \varphi)\big]\nonumber\\
&\quad \ +\delta^\alpha_\beta\big[-v^\nu v^\gamma \varphi_{,\nu\vert \gamma} +(\p_v \varphi)^2-g_v(v,v)g^{\mu \nu} \varphi_{,\mu} \varphi_{,\nu} \big] \label{ful}  \\
&\quad \ +g_v(v,v)\Big[-2 L^{\alpha \delta}_\beta \varphi_{,\delta}  -g^{\alpha \delta} \varphi_{,\delta \vert \beta}-C^{\alpha \delta}_\beta v^\gamma \varphi_{,\delta\vert \gamma}+g^{\alpha \delta} \varphi_{,\delta} \varphi_{,\beta} \nonumber \\
&\quad \ -2 (\p_v \varphi) C^{\alpha \delta}_\beta \varphi_{,\delta}  -g_v(v,v)C^{\nu \mu}_\beta C^{\alpha \delta}_\nu \varphi_{,\mu} \varphi_{,\delta}-g_v(v,v) g^{\nu \mu} \varphi_{,\mu} \frac{\p C^{\alpha \delta}_\beta}{\p v^\nu} \,\varphi_{,\delta}\Big] \nonumber
\end{align}
 %
where $\vert$ is the horizontal covariant derivative with respect to the Cartan (Chern-Rund) connection of $\mathscr{L}$.

In this expression there are only terms proportional to $v^\alpha$, $v_\beta$, $\delta^\alpha_\beta$ or $g(v,v)$. In Sect.\ \ref{ral} we considered the case $g(v,v)=0$  (there $v$ was denoted $n$)  and we restricted the action of $R_v$ to a subspace annihilated by $v_\beta$. Furthermore, we  passed to a quotient, by regarding two vectors as equivalent whenever they differ by a term proportional to $v$. In that context, the previous expression shows that  $\delta \bar{R}\propto Id$, and hence that its traceless part is left invariant by the conformal transformation. We conclude
\begin{theorem}
The endomorphism $\bar{C}$ of Sect.\ \ref{ral} is conformally invariant.
\end{theorem}
The existence of this conformal invariant is quite remarkable  since
there is no known generalization of the Weyl conformal curvature tensor to (pseudo-) Finsler geometry \cite{knebelman29b,rund59,hashiguchi76,ikeda97,hojo00}. Of course, $\bar{C}$ cannot be defined in the positive definite case since  $v$ must be lightlike.

 Let us  set
\[I_\gamma= C^\alpha_{\alpha \gamma}\qquad  \textrm{ and }\qquad  J_\gamma= L^\alpha_{\alpha \gamma} , \]
 and recall the identity \cite[Eq.\ (52)]{minguzzi14c} $J_\beta =v^\gamma\nabla_\gamma^{HC} I_\beta$. Clearly, $C^\alpha_{\beta \gamma}$ and $I_\gamma$ are conformally invariant.
Taking the trace of (\ref{ful}) we obtain
\begin{align*}
\tilde{Ric}(v)&= Ric(v)-(n-1) v^\beta v^\alpha \ln \Omega_{, \beta\vert \alpha}-g_v(v,v) g^{\mu \nu}\ln \Omega_{,\mu \vert \nu}+(n-1)(\p_v \ln \Omega)^2\\
& \qquad -(n\!-\!1) g_v(v,v) g^{\mu \nu} \ln \Omega_{,\mu} \ln \Omega_{, \nu}-g_v(v,v)\Big\{2 \ln \Omega_{, \mu} J^\mu+I^\mu v^\nu \ln \Omega_{, \mu\vert \nu}\\
&\!\!\!\!\!\!\!\!\!\!\!\!\!\!\!\!\!\!\!\! +g_v(v,v)g^{\nu \mu} \frac{\p I^\delta}{\p v^\nu} \ln  \Omega_{,\delta} \ln \Omega_{,\mu} +g_v(v,v) C^{\alpha \mu}_{\gamma} C^{\gamma \nu }_{\alpha} \ln  \Omega_{,\mu} \ln \Omega_{,\nu}+2(\p_v \ln \Omega) I^\delta  \ln \Omega_{, \delta}\Big\}.
\end{align*}


The transformation rule for the Ricci scalar can be expressed in terms of the  horizontal Chern-Rund derivative of $\tilde{\mathscr{L}}$, denoted $\tilde{\vert}$ for short, by inverting the roles of $\tilde{\mathscr{L}}$ and $\mathscr{L}$ in the previous formula
\begin{align}
\tilde {Ric}(v)&= {Ric}(v)-\!(n\!-\!1) \Omega^{-1} \Omega_{,\alpha \tilde{\vert}\beta} v^\alpha v^\beta
+2 \tilde{\mathscr{L}}\{-\Omega^{-1}\Omega_{,\gamma \tilde{\vert} \delta} +n\Omega^{-2}\Omega_{,\gamma} \Omega_{,\delta} \} \tilde{g}_v^{\gamma\delta} \nonumber \\
&\qquad +2 \tilde{\mathscr{L}}\big\{ -2 \ln \Omega_{, \mu} \tilde{J}^\mu- \Omega^{-1}\tilde{I}^\mu v^\nu \Omega_{, \mu \tilde{\vert} \nu}+3 (\p_v \ln  \Omega) \tilde I^{\delta} \ln \Omega_{, \delta} \nonumber\\
& \qquad+2 \tilde{\mathscr{L}} \tilde{g}^{\nu \mu} \frac{\p \tilde{I}^\delta}{\p v^\nu} \ln \Omega_{,\delta} \ln \Omega_{,\mu}+2 \tilde{\mathscr{L}} \tilde{C}^{\alpha \mu}_\gamma \tilde{C}^{\gamma \nu}_\alpha \ln \Omega_{, \mu} \ln \Omega_{, \nu}  \big\} \ . \label{nua}
\end{align}
This formula will be useful in the next section.

\section{Asymptotics, singularities, differentiability}

Not so easy generalizations are all those which involve the curvature since, as we mentioned, there are several notable Finsler connections on a Finsler spacetime. Fortunately, for what concerns singularity theorems, curvature comes into play through the Raychaudhuri equation and hence only through the curvature of the non-linear connection.

Since some singularity theorems involve a notion of ``weakly asymptotically simple and empty spacetime'' or ``asymptotically flat'' spacetime it is necessary to define these notions in the Finsler case.

\begin{definition}
A Finsler spacetime $(M,\mathscr{L})$ is said to be {\em asymptotically simple} if there is a strongly causal Finsler spacetime $(\tilde M,\tilde{\mathscr{L}})$ and an imbedding $\theta\colon M\to \tilde M$, which imbeds $M$ as a manifold with smooth boundary $\p M$ in $\tilde M$, such that
\begin{itemize}
\item[(1)] there is a smooth (say $C^3$ at least) function $\Omega\colon\tilde M\to \mathbb{R}$ such that on $\theta(M)$, $\Omega$ is positive and $ \theta^* (\Omega^2) g=\theta^* \tilde g$,
\item[(2)] on $\p M$, $\Omega=0$ and $\dd \Omega \ne 0$,
\item[(3)] every null geodesic in $M$ has two endpoints in $\p M$,
\end{itemize}
further we say that it is {\em  asymptotically simple and empty} if it satisfies also\footnote{Physically, it is not clear whether $v$ should be better restricted to the future causal cone. }
\begin{itemize}
\item[(4)] $Ric(v)=0$ in an open neighborhood of $\p M$ in $\tilde M$.
\end{itemize}
\end{definition}


By Proposition \ref{sin} the geodesics reaching some point $p\in \p M$ transversally are necessarily complete as can be easily shown Taylor expanding $\Omega$ at $p$.


Multiplying Eq.\ (\ref{nua}) by $\Omega^2$ and noticing that $\tilde {Ric}$, $\tilde C$, $\tilde I$, $\tilde J$ are $C^1$ at $\p M$ where $\Omega=0$, we get for every non-lightlike $v\ne 0$, and hence, by continuity, for every $v\ne 0$
\begin{equation} \label{aax}
n\Omega_{,\gamma} \Omega_{,\delta}  \tilde{g}_v^{\gamma\delta}+2 \tilde{\mathscr{L}} \tilde{g}^{\nu \mu} \frac{\p \tilde{I}^\delta}{\p v^\nu}  \Omega_{,\delta} \Omega_{,\mu}+2 \tilde{\mathscr{L}} \tilde{C}^{\alpha \mu}_\gamma \tilde{C}^{\gamma \nu}_\alpha  \Omega_{, \mu}  \Omega_{, \nu} +3 (\p_v  \Omega) \tilde I^{\delta}  \Omega_{, \delta}  =0.
 \end{equation}
 This equation implies that $\ker \dd \Omega$ is a null hyperplane, that is, tangent to the Finsler future causal cone, and hence that $\p M$ is a null hypersurface. To see this let $v^\pm$ be the Legendre transform of $\pm \dd \Omega$: $g_{v^{\pm}}(v^\pm, \cdot)=\pm \dd \Omega$. Let $v=v^{\pm}$ in Eq.\ (\ref{aax})
 and observe that $\pm \tilde{g}^{\nu \mu} \Omega_{,\mu} =v^\nu$. Since  $\tilde{I}^\delta$ is positive homogenous of zero degree in $v$,  the second term vanishes. Similarly, the third and fourth terms vanish since $v^\mu$ annihilates the Cartan torsion. We are left with $\Omega_{,\gamma} \Omega_{,\delta}  \tilde{g}_v^{\gamma\delta}=0$ which states that $v^+$ and $v^-$ are lightlike. One of them is actually future directed \cite[Prop.\ 9]{minguzzi13c}.


Since $\p M$ is a null hypersurface, $M$ lies locally to the past or future of it. Thus $\p M$ consists of two disconnected components: $\mathscr{I}^+$ on which every null geodesic has its future endpoint, and  $\mathscr{I}^{-}$ on which every null geodesic has its past endpoint.

Now the usual theorems follow since proofs extend word for word from their Lorentzian versions \cite{hawking73} as they use topological arguments which do not depend on algebraic details. Thus

\begin{proposition}
Every four dimensional asymptotically simple and empty space is globally hyperbolic; furthermore, $\mathscr{I}^+$ and $\mathscr{I}^-$ are topologically $\mathbb{R}^1 \times S^{2}$, and $M$ is $\mathbb{R}^{4}$.
\end{proposition}

Since the notion of  asymptotically simple and empty spacetime is too restrictive we introduce as usual a modification.
We say that $(M,\mathscr{L})$ is {\em weakly asymptotically simple and empty} (WASE) if there is an asymptotically simple and empty space $(M',\mathscr{L}')$ and a neighborhood $U'$ of $\p M'$ in $M'$ such that $U'\cap M'$ is isometric to an open set $U$ of $M$.

These results should convince the reader that asymptotic concepts can be treated quite successfully in analogy with the Lorentzian case.

With these preliminaries and the results on the Raychaudhuri equation developed in the first sections the next result becomes clear (observe that it includes those theorems which use the averaged convergence conditions; their definition is clear and will not be repeated here).


\begin{proposition}
The singularity theorems by Penrose \cite{penrose65b}, Hawking \cite{hawking66,hawking67,hawking73}, Hawking and Penrose \cite{hawking70}, Geroch \cite{geroch67}, Gannon \cite{gannon75,gannon76},  Tipler \cite{tipler77,tipler78}, Borde \cite{borde87,borde94}, Kriele \cite{kriele89,kriele90},  and the author \cite{minguzzi07d}, but also  Friedman, Schleich  and Witt's Topological censorship theorem\cite{friedman93}, or results on the simply connectedness of the domain of outer  communication \cite{chrusciel94b}, or on the spherical shape of Black Holes \cite{galloway96b}, generalize word for word to the Finslerian domain.
\end{proposition}

The proofs coincide word for word with the Lorentzian ones. In fact so many results generalize, e.g.\  also \cite{joshi81,joshi87}, that not all references can be included as the bibliography would become too long.

\section{Conclusions}

We have shown that many advanced results of causality theory including singularity theorems generalize to the Finslerian spacetime case. While it is intuitive that most topological arguments should generalize, one has to check more closely all those arguments which involve the curvature or tricky algebraic calculations. In fact, several connections exists in Finsler geometry, so  the generalization is not always uniquely determined or straightforward and, moreover, the curvatures do not satisfy all the symmetries which share their pseudo-Riemannian analogs (when such analogs can be identified).

Fortunately, singularity theorems involve the study of geodesic congruences, which are expected to depend only on the spray and its derived curvatures. In fact, we observed that only the curvature of the non-linear connection was really involved in the calculations. This fact meant a radical simplification since we didn't have to make a choice of {\em physical Finsler connection}, as the Finsler (linear) connection was only used as a tool in proofs.

Using this approach  we have been able to show that the Raychaudhuri equation generalizes, as do  its consequences for chronality. From here it becomes easy to check that the notable singularity theorems of Lorentzian geometry extend to the Lorentz-Finsler case word for word, both in statement and in proof. We have also included further results on Finsler causality which might help the reader to gain some confidence and familiarity with this theory.

In conclusion, the Lorentz-Finsler theory accomplishes and preserves results of physical significance that no other generalization of general relativity could possibly claim. This fact selects this theory as a serious  candidate for a modified gravitational theory.


\section*{Acknowledgments}   I thank E.\ Caponio for a useful comment on Beem et al.\ theorem on chronality \cite[Theor.\ 10.72]{beem96} . This work has been partially supported by GNFM of INDAM. 

\end{document}